\documentclass[submission,copyright,creativecommons]{eptcs}
\providecommand{\event}{GandALF 2017} 
\usepackage{underscore}           
\usepackage{amsmath}

\usepackage{tikz}
\usetikzlibrary{automata,positioning,calc,decorations.text,shapes,fit,arrows,topaths}

\usepackage{paralist}

\usepackage{semantic}

\usepackage{esvect}

\usepackage{times}

\usepackage{microtype}

\usepackage{amssymb}

\usepackage{amsthm}

\newtheorem{definition}{Definition}
\newtheorem{lemma}{Lemma}
\newtheorem{theorem}{Theorem}
\newtheorem{example}{Example}
\newtheorem{assumption}{Assumption}
\usepackage{wrapfig}

\newcommand{\QED}{\hfill\ensuremath{\square}}

\newcommand{\ppf}{\ensuremath{\mathcal{PPF}}}

\newcommand{\sn}{\ensuremath{\mathcal{SN}}}
\newcommand{\ppl}{\ensuremath{\mathcal{PPL}}}
\newcommand{\kbl}{\ensuremath{\mathcal{KBL}}}
\newcommand{\sat}{\ensuremath{\vDash}}
\newcommand{\der}{\ensuremath{\vdash}}

\newcommand{\poval}{\ensuremath{\pi}}
\newcommand{\propSet}{\ensuremath{\mathcal{P}}}
\newcommand{\poSet}{\ensuremath{\Pi}}
\newcommand{\conSet}{\ensuremath{\mathcal{C}}}
\newcommand{\struct}{\ensuremath{\mathcal{A}}}

\newcommand{\kb}{\ensuremath{\mathit{KB}}}

\newcommand{\wkbl}{\ensuremath{\mathcal{F_{KBL}}}}

\newcommand{\actSet}{\ensuremath{\Sigma}}
\newcommand{\etal}{\mbox{\emph{et al.\ }}}

\newcommand{\nat}{\ensuremath{\mathbb{N}}}

\newcommand{\mlt}{\ensuremath{\mathcal{M}^{lt}}}

\newcommand{\lang}{\ensuremath{\mathcal{L}}}
\newcommand{\aguni}{\ensuremath{\mathcal{AU}}}

\newcommand{\voca}{\ensuremath{\mathcal{T}}}
\newcommand{\dom}{\ensuremath{\mathit{dom(\mathcal{A})}}}

\newcommand{\kripkeval}{\ensuremath{\pi}}
\newcommand{\accrelation}{\ensuremath{\mathcal{K}}}

\newcommand{\snv}{\ensuremath{\mathit{SN}}}
\newcommand{\alice}{\ensuremath{\mathit{Alice}}}
\newcommand{\bob}{\ensuremath{\mathit{Bob}}}
\newcommand{\charlie}{\ensuremath{\mathit{Charlie}}}
\newcommand{\audience}{\ensuremath{\mathit{Audience}}}
\newcommand{\owner}{\ensuremath{\mathit{owner}}}
\newcommand{\loc}{\ensuremath{\mathit{loc}}}
\newcommand{\location}{\ensuremath{\mathit{loc}}}
\newcommand{\post}{\ensuremath{\mathit{post}}}

\newcommand{\Blocked}{\ensuremath{\mathit{Blocked}}}
\newcommand{\blocked}{\ensuremath{\mathit{blocked}}}
\newcommand{\friendRequest}{\ensuremath{\mathit{friendRequest}}}
\newcommand{\friend}{\ensuremath{\mathit{friend}}}
\newcommand{\friends}{\ensuremath{\mathit{friends}}}
\newcommand{\Friend}{\ensuremath{\mathit{Friend}}}

\newcommand{\age}{\ensuremath{\mathit{age}}}
\newcommand{\pub}{\ensuremath{\mathit{pub}}}
\newcommand{\library}{\ensuremath{\mathit{library}}}
\newcommand{\outerk}{\ensuremath{\mathit{outerK}}}

\newcommand{\ag}{\ensuremath{\mathit{Ag}}}
\newcommand{\phisnv}{\ensuremath{\phi_{\snv}}}

\newcommand{\Phisnv}{\ensuremath{\Phi_{\snv}}}
\newcommand{\mPhisnv}{\ensuremath{\Phi^{m}_{\snv}}}
\newcommand{\phikbi}{\ensuremath{\phi_{\kbi}}}

\newcommand{\sub}{\ensuremath{\mathit{Sub}}}
\newcommand{\subplus}{\ensuremath{\mathit{Sub}^{+}}}
\newcommand{\con}{\ensuremath{\mathit{Con}}}

\newcommand{\kt}{\ensuremath{\mathcal{KT}}}

\newcommand{\ktm}{\ensuremath{\mathcal{KT}^{m}}}

\newcommand{\kbi}{\ensuremath{\kb_i}}

\newcommand{\co}{\ensuremath{\mathit{co}}}

\newcommand{\kaxioma}{\textbf{K}}
\newcommand{\kdaxioma}{\textbf{KD}}
\newcommand{\kdfouraxioma}{\textbf{KD4}}

\newcommand{\cc}{\textbf{cc}}

\let\temp\phi
\let\phi\varphi
\let\varphi\temp

\title{Model Checking Social Network Models}
\def\titlerunning{Model Checking Social Network Models}

\author{
  Ra\'ul Pardo \qquad Gerardo Schneider
  \institute{
    Department of Computer Science and Engineering,\\
    Chalmers $\mid$ University of Gothenburg, Sweden.
  }
  \email{pardo@chalmers.se \quad \qquad gerardo@chalmers.se}
}
\def\authorrunning{R. Pardo \& G. Schneider}

\begin{document}

\maketitle

\begin{abstract}
  A {\em social network service} is a platform to build social
  relations among people sharing similar interests and activities.
  The underlying structure of a social networks service is the
  {\em social graph}, where nodes
  represent users and the arcs represent the users' social
  links and other kind of connections. One important concern in social networks
  is {\em privacy}: what others are (not) allowed to {\em know} about
  us.  The ``logic of knowledge'' (\emph{epistemic logic}) is thus a
  good
formalism to define, and reason about, privacy
  policies.
In this paper we consider the problem of verifying knowledge properties
over {\em social network models} (SNMs), that is social graphs enriched with {\em knowledge bases} containing the information that the users know.
More concretely, our contributions are:
i) We prove that the model checking problem for
epistemic properties
over SNMs is decidable;
ii) We prove that a number of properties of knowledge that are sound w.r.t.~Kripke models are also sound w.r.t.~SNMs;
iii) We give a satisfaction-preserving encoding of SNMs into {\em canonical} Kripke models, and we also characterise which Kripke models may be translated into SNMs;
iv) We show that, for SNMs, the model checking problem is cheaper than the one based on standard Kripke models.
Finally, we have developed a proof-of-concept implementation of the model-checking algorithm for SNMs.
\end{abstract}


\section{Introduction}
\label{sec:introduction}

{\em Social networks services} (or simply {\em social networks}) are one of the most popular services on the Internet nowadays.
One of the main concerns in social networks is that of privacy: most users are not in full control over what they share, and it is not uncommon that private and personal data is leaked to an unintended audience \cite{YKBA11afps+}. These concerns arise because users cannot determine (in a precise manner) who {\em knows} their personal information.
One solution is to provide users with more fine grained control over who knows their information. Epistemic logic or ``the logic of knowledge''~\cite{FHM+95rk} offers great precision and granularity for modelling and reasoning about the knowledge of the (users or {\em agents}) in a system.


In \cite{PS14fpp} we introduced \ppf, a formalism based on epistemic logic to specify privacy policies in social networks, and to enable a formal assessment on whether these policies are preserved.
\ppf\ consists of:
i) A generic model for social networks (SNMs);
ii) A knowledge-based logic (\kbl) to reason about the social network and privacy policies;
iii) A formal language (\ppl) to describe privacy policies (based on \kbl).
In \cite{PardoLic}, \ppf\ was further extended by
providing agents with a deductive engine to perform knowledge inferences, and including an operational semantics to model the dynamics of social networks.


\ppf~has been specifically designed for privacy policies for real social networks, and that is why the language \ppl~and the underlying logic \kbl~are interpreted over SNMs and not over Kripke models (\emph{possible-worlds} semantics), which is the ``standard'' way to give semantics to epistemic logic.
In Kripke models the uncertainty of the agents is modelled using an {\em accessibility relation}. This relation connects all the worlds in the model that an agent considers possible. If a formula is true in all of them, then the agent knows it.
This does not correspond to the way users in real world social networks acquire and reason about information.
Typically, when a user joins a social network, she knows none or a few facts about it. The system might suggest some friends that are retrieved from the user's phone contacts. As the user makes new friends and they share information, her knowledge starts to grow, and later from this set of accumulated knowledge users may derive new facts.

There are two main advantages in \ppf's design (as opposed to standard Kripke models):

\begin{compactenum}
\item \textbf{It preserves the original structure of real social networks}. The models in \ppf~(SNMs) consist of the {\em social graph}~\cite{K14cn} and a knowledge base per user. The topology of the social graph provides information regarding the relationships between users (e.g., friends, colleagues,...). The knowledge base gives semantics to the modality $K_i \phi$ (user $i$ knows $\phi$).
Knowledge bases are not a new invention, they are just an instance of the {\em syntactic} approach to modelling knowledge~\cite{HC13bsal}.
This structure is also important from the enforcement point of view since it facilitates the integration of the framework with the target social network.

\item \textbf{Checking  whether a user knows something must be as efficient as possible}. The privacy policies that users can specify in \ppf~talk about knowledge, e.g., ``Only my friends can know my location'' or ``Only my family can know that I am going to my father's birthday party''. Therefore, the enforcement of \ppf~privacy policies mainly depends on how efficiently these checks are performed. Social networks have millions of users, who disclose tons of information per second. As a consequence, a slow enforcement mechanism would not work in practice. By splitting the users' knowledge in different knowledge bases, the complexity of checking whether a user knows a piece of information can be significantly reduced. In Section~\ref{sec:complexity} we study the improvement in complexity of having separated knowledge bases as opposed to standard Kripke semantics.
\end{compactenum}


The properties of knowledge related to human reasoning, present in Kripke models, have been studied for decades and they are well-understood \cite{FHM+95rk}. On the other hand, the properties of knowledge in SNMs have not been throughly studied. Therefore, several questions need to be answered:
\begin{inparaenum}[i)]
\item What is the relation between SNMs and Kripke models?
\item Does this slightly different representation of knowledge preserve the same properties?
\item Is it possible to determine whether an epistemic formula written in \kbl~is satisfied on a given SNM?\footnote{Answering this question will also solve the model checking problem for privacy policies written in \ppl, as checking conformance of \ppl~is reduced to checking satisfaction of a \kbl~formula.}
\end{inparaenum}
In this paper we study in depth the answer to these questions
providing evidence that
\ppf~not only offers advantages from the practical point of view, but also models knowledge as traditionally understood and accepted in the epistemic logic literature.


More concretely, our contributions are:
\begin{inparaenum}[i)]
\item A proof that model checking \kbl~formulae over SNMs is decidable, the algorithm being an implementation of the satisfaction relation for \kbl~(Section \ref{sec:mc});
\item A logical characterisation of a number of properties of knowledge for SNMs including {\it common} and {\it distributed knowledge} (Section \ref{sec:axioms}).
\item A translation from SNMs into {\em canonical} Kripke models, together with a proof that satisfaction is preserved (Section~\ref{sec:canonical});
we also show that it is always possible to reconstruct the original SNM from the canonical Kripke model, by considering the state associated with the characteristic formulae (Section~\ref{sec:kripketosnm});
\item A formal comparison of the complexity of the model checking problem for SNMs and for Kripke models where we show that the former is more efficient (Section~\ref{sec:complexity}).
\end{inparaenum}
Additionally, we provide a proof-of-concept implementation of the model-checking algorithm.\footnote{\url{https://github.com/raulpardo/kbl-model-checker}}
The extended version of this paper includes the proofs of all Theorems and Lemmas~\cite{appendix}.


\section{Preliminaries}
\label{sec:preliminaries}

Here we briefly recall First-Order Epistemic Logic \cite{FHM+95rk}, social network models and the logic \kbl~\cite{PardoLic}.

\subsection{First-Order Epistemic Logic}
\label{sec:epistemicfol}
We start with a set \voca, consisting of \emph{relation symbols}
($p$), \emph{function symbols} ($f$) and \emph{constants symbols}
$(c)$. Hereafter we will refer to \voca~as the \emph{vocabulary}. Each
relation and function symbol has an implicit {\em arity}  which
corresponds to the number of arguments it takes. Function and
relation symbols are interpreted over elements of a domain.
We assume an infinite supply of variables, which we write as $x, y$ and so on.
We can form terms using constants, variables, and function symbols.
Formally, a {\em term} $t$ is recursively
defined as follows: $t \mbox{ ::= } c \; | \; x \; | \; f(\vv{t})$,
where $\vv{t}$ represents a list of terms $t_1, \ldots, t_k$. An
\emph{atomic formula} is of the form $p(\vv{t})$ where $p$ is a
relation symbol. Let $\ag$ be a set of {\em agents}, $i \in \ag$ and
$G \subseteq \ag$, the syntax of {\em First-Order Epistemic Logic} (FOEL), denoted as \lang, is
recursively defined as follows \cite{FHM+95rk}:
\[ \phi \mbox{::=} p(\vv{t}) \; | \; \phi \wedge \phi \; | \; \neg \phi \; | \; \forall x. \phi \; | \; K_i \phi \; \]
The remaining epistemic modalities are defined as $S_G \phi \triangleq \bigvee_{i \in G}  K_i \phi$ and $E_G \phi \triangleq \bigwedge_{i \in G} \phi$.
The intuitive meaning of the modalities is the following:
\begin{inparaenum}
\item[$K_{i} \phi$,] agent $i$ knows $\phi$;
\item[$E_{G} \phi$,] everyone in the group $G$ knows $\phi$;
\item[$S_{G} \phi$,] someone in the group $G$ knows $\phi$.
\end{inparaenum}
The semantics of FOEL formulae is given using \emph{relational Kripke models}.
In what follows we sometimes omit relational and write Kripke models.

\begin{definition}[\cite{FHM+95rk}]
  A \emph{relational Kripke Model} is a tuple of the form $\langle S,
  \kripkeval, \{\accrelation_i\}_{i \in \ag} \rangle$, where:
  \begin{compactitem}

  \item $S$ is a non-empty set of {\em states} (or {\em worlds}).

  \item $\kripkeval : S \rightarrow \struct$ is a function that associates to each world a {\em relation structure} for a fixed vocabulary \voca. As usual,~\struct~consists of a domain \dom, an assignment of a k-ary relation $P^{\struct} \subseteq \dom^{k}$ for each relation symbol, an assignment of a k-ary function $f^{\struct}: \dom^{k} \rightarrow \dom$ for each function symbol and an assignment of a member $c^{\struct}$ of the domain for each constant symbol. 

  \item $\{\accrelation_i\}_{i \in \ag}$ where $\accrelation_i \subseteq S \times S$ is an {\em accessibility relation} between states.
  \end{compactitem}
\end{definition}

\begin{example}
  Let us consider a Kripke structure consisting of agents $a$ and $b$, states $s_0$, $s_1$ and $s_2$, a predicate $p$ with arity 1 and relations $\mathcal{K}_a = \{(s_0, s_1), (s_1, s_0)\}$ and $\mathcal{K}_b = \{(s_1, s_2), (s_2, s_1)\}$. We assume here that all relational structures $\pi(s_n)$ have a common domain $\dom = \{a,b\}$, i.e., $\ag$. Moreover, $a \in P^{\pi(s_0)}$ and $a \in P^{\pi(s_1)}$.  Fig.~\ref{fig:ksexample} shows a graphical representation of the described model.\QED
\end{example}


\begin{figure}[!t]
  \centering
  \footnotesize
      \begin{tikzpicture}[>=stealth',every text node part/.style={align=center},on grid,auto, semithick]
	\node (s0) [circle, draw, minimum width=10pt, minimum height=8ex, label=below:$s_0$]%
	      {$p(a)$};
        \node (s1) [circle,draw,minimum width=10pt, minimum height=8ex, label=below:$s_1$]%
		    [right = 3cm of s0]%
		    {$p(a)$};
	\node (s2) [circle,draw,minimum width=10pt, minimum height=8ex,label=below:$s_2$]%
		        [right = 3cm of s1]%
		        {};
	\path[-] [] (s0) edge node {$a$} (s1);
	\path[-] [] (s1) edge node {$b$} (s2);
      \end{tikzpicture}
  \caption{\label{fig:ksexample} Relational Kripke structure}
\end{figure}
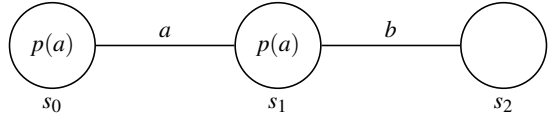

Usually free variables and terms are interpreted using a {\em valuation} function, which is parametrised with a relational structure depending of the state of the Kripke model in which the
formula is evaluated.
For simplicity, in this paper we will assume that formulae in \lang~do not contain free variables (i.e., all variables are quantified) and the interpretation of functions and constants is the same independently of the  state where they are evaluated.
Thus, we assume that terms are implicitly interpreted and we do not include the valuation function as a parameter in the satisfaction relation below.


\begin{definition}[\cite{FHM+95rk}]
  Given a non-empty set of agents $\ag$, a relational Kripke model $M$, a state $s \in M$, agents $i,j,u \in \ag$ and a finite set of agents $G \subseteq \ag$ 
  , we define what it means for $\phi \in \lang$ 
  to be {\em satisfied} by $(M, s)$, written $(M, s) \sat \phi$, as shown in Table \ref{tab:kripkesat}.
\end{definition}

\begin{table}[!t]
  \centering
  {
  \begin{tabular}{l c l}
    $(M, s) \sat p(t_1,\ldots,t_k)$ & \mbox{iff} & $(t_1,\ldots,t_k) \in P^{\pi(s)}$\\
    $(M, s) \sat \neg \phi $ & \mbox{iff} & $(M, s) \not \sat \phi$ \\
    $(M, s) \sat \phi_1 \wedge \phi_2 $ & \mbox{iff} & $(M, s) \sat \phi_1 \mbox{ and } (M, s) \sat \phi_2$\\
    $(M, s) \sat \forall x. \phi$ & \mbox{iff} & \mbox{ for all  } $v \in \mathit{dom}(\pi(s))$, $(M, s) \sat \phi[v/x]$ \\
    $(M, s) \sat K_i \phi$ & \mbox{iff} & $(M, t) \sat \phi~\mbox{for all $t$ such that}~(s,t) \in \mathcal{K}_i$\\
  \end{tabular}}
  \caption{Satisfaction relation over Kripke models\label{tab:kripkesat}}
\end{table}

We say that a formula $\phi$ is {\em valid in a Kripke model} $M$, and we write $M \sat \phi$, if $\forall s \in M \; (M, s) \sat \phi$. Moreover, we say that $\phi$ is {\em valid}, denoted as $\sat \phi$, if for all Kripke models $M$ it holds $M \sat \phi$.

\begin{example}
  Let $M$ be the model presented in Fig.~\ref{fig:ksexample}.
  It holds that $(M, s_0) \sat K_a p(a)$, since $p(a)$ holds in $s_0$ and in all the states accessible for $a$ from $s_0$ (only $s_1$).
  It also holds that $(M, s_1) \sat \neg K_b p(a)$, since in one of the states that $b$ considers possible $p(a)$ is not true.
  In particular, $(M, s_2) \sat \neg p(a)$.
\end{example}

\subsection{\kbl~and Social Network Models}
\label{sec:ppf}

\kbl~is a {\em knowledge-based logic} for social networks.
It contains all the knowledge modalities presented in \lang, and additionally, it includes two special types of predicates.
The connection and action predicates.
Connection predicates represent the ``social'' connections between users.
For instance, friends, colleagues, family, co-workers, and so forth.
Action predicates model the permitted actions a user may execute.
For example, Alice can send a friend request to Bob or Alice can join events created by Bob.
Note that action predicates are not deontic modalities.
Hereafter we use $\conSet$ and $\actSet$ to denote sets of indexes for connections and permissions, respectively.
As before the set \ag~represents a set of agents in the system.

\begin{definition}[]\label{kbl:syntax}
  Given $i,j \in \ag$, a set of predicate symbols $\propSet$ such that $a_n(i,j), c_m(i,j), p(\vv{t}) \in \propSet$ where $m \in \conSet$ and $n \in \actSet$, and $G\subseteq Ag$, the syntax of the {\em knowledge-based logic} \kbl~is inductively defined as:
  \begin{eqnarray*}
    \begin{tabular}{l l l}
      $\phi$ & \mbox{::=} & $c_m(i,j)\; | \; a_n(i,j) \; | \; p(\vv{t}) \; | \; \phi \wedge \phi \; | \; \neg \phi \; | \; \forall x. \phi \; | \; K_i \phi \; $\\
    \end{tabular}
  \end{eqnarray*}
  As before, the remaining epistemic modalities are defined as $S_G \phi \triangleq \bigvee_{i \in G}K_i \phi$ and $E_G \phi \triangleq \bigwedge_{i \in G} \phi$.
\end{definition}

Terms and atomic formulae are defined as for \lang.
\wkbl~denotes the set of {\em well-formed formulae} of \kbl~(category $\phi$ of Def.~\ref{kbl:syntax}).

Social networks are usually modelled as graphs where nodes represent the users (or agents), and edges represent different relationships among agents or any other social network specific information \cite{K14cn}.  These graphs are known as \emph{social graphs}.
Here we enrich social graphs with information about the agents knowledge, permissions, connections and privacy policies as defined below.

\begin{definition}[Social Network Model]\label{def:sn}
  Given
  a set of \kbl~formulae $\mathcal{F}$,
  a set of privacy policies $\poSet$,
  and a finite set of agents $Ag \subseteq \aguni$ from a universe $\aguni$,
  a {\em social network model (SNM)} is a social graph of the form $\langle \ag, \struct, \kb, \poval \rangle$, where

  \begin{compactitem}
  \item $\ag$ is a nonempty finite set of {\em nodes} representing the agents of the social network.

  \item $\struct$ is a first-order relational structure for the fixed vocabulary of the {\em SNM}, which as before, consists of a finite domain \dom\footnote{For the sake of clarity in definitions and proofs and w.l.o.g. we have only considered a single finite domain in the formal definition. However, in the rest of the paper we will assume that we have a finite set of finite domains. For instance, we can have \dom~consisting of the domain of agents, timestamps, indexes for pictures, etc. All the results also hold in SNMs consisting of multiple domains as we consider a finite number of finite domains.}, an assignment of a k-ary relation $P^\struct \subseteq \dom^\struct$ for each predicate symbol, an assignment of a k-ary $f^\struct: \dom^k \rightarrow \dom$ for each function symbol and assignment of a member $c^\struct$ of the domain for each constant symbol.

  \item $\kb: \ag \rightarrow 2^{\mathcal{F}}$ is a function that returns a finite set of accumulated knowledge for each agent, stored in what we call the {\em knowledge base} of the agent. We write $\kb_i$ to denote $\kb(i)$. 

  \item $\poval: \ag \rightarrow 2^{\poSet}$ is a function that returns a finite set of privacy policies for each agent. We write $\poval_i$ to denote $\poval(i)$.
  \end{compactitem}
\end{definition}

\begin{figure}[t]
  \centering
  \scalebox{0.75}{
      \begin{tikzpicture}[>=stealth',every text node part/.style={align=center},on grid,auto, minimum height=6ex,semithick]
	\node (alice) [rectangle, rounded corners, draw, minimum width=200pt, label=above:Alice]%
	      {$\post(\bob,\pub,1)$\\
                $\forall t.(\post(\bob,\pub,t) \implies \loc(\bob,\pub,t))$};
        \node (bob) [rectangle, rounded corners, draw,minimum width=200pt,label=left:Bob]%
		    [right = 10cm of alice, yshift=-1.07cm]%
                    {~};
	\node (charlie) [rectangle, rounded corners, draw,minimum width=200pt,label=below:Charlie]%
		        [below = 2cm of alice]%
		        {$\post(\bob,\library, 2)$};
	\draw[<->] [dashed] (alice) -| node {$\Friend$} (bob);
	\draw[->] [dashed] (bob) |- node {$\Blocked$} (charlie);
	\draw[->] [dotted] (charlie) -- (-5,-2) |- node {$\friendRequest$} (alice);
      \end{tikzpicture}}
  \caption{\label{fig:snmexample} Example of Social Network Model}
\end{figure}
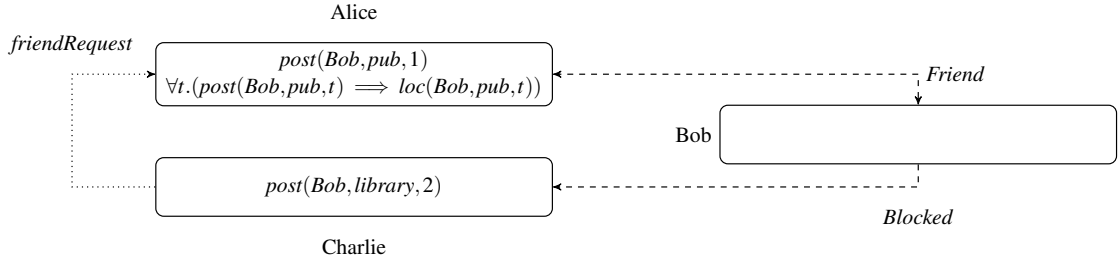

The shape of the relational structure $\struct$ depends on the concrete the social network.
Connections and permission actions between agents, i.e., edges of the social graph, are represented
as families of binary relations, $ \{C_i\}_{i \in  \conSet} \subseteq 2^{\ag \times \ag} $ and $ \{A_i\}_{i \in \actSet} \subseteq 2^{\ag \times \ag} $ over the domain of agents. Sometimes, we write an atomic formula, e.g.~$\friends(a,b)$ to denote that the elements $a, b \in \ag$ belong to a binary relation, $\friends$, defined over pairs of agents as expected.
\sn~denotes the universe of all possible SNMs.

The knowledge base $\kb_i$ of each agent $i$ contains the explicit knowledge that the agent has.
Besides this explicit knowledge, agents also know anything that can be derived from formulae in their knowledge bases (using the \kdfouraxioma~axiomatisation of epistemic logic~\cite{FHM+95rk}).

\begin{definition}
  A {\em derivation} of a formula $\phi \in \wkbl$, is a finite sequence of formulae $\phi_1, \ldots,$ $\phi_n = \phi$ where each $\phi_i$, for $1 \leq i \leq n$, is either an instance of the axioms or the conclusion of one of the derivation rules of the \kdfouraxioma~axiomatisation which premises have already been derived, i.e., it appears as $\phi_j$ with $j < i$.

  Given a set of formulae $\Gamma \in 2^\wkbl$, we write $\Gamma \der \phi$ to denote that $\phi$ can be derived from $\Gamma$.
  \label{def:derivation}
\end{definition}

Additionally, we impose two assumptions in users' knowledge bases:
\begin{compactenum}[i)]
  \item $\phi$ and $\neg \phi$ cannot be derivable in the same $\kbi$. It prevents users from having inconsistent knowledge.
  \item If $\phi$ is in $i$'s knowledge base, $K_i \phi$ is also there. In this way we make users aware of their knowledge.
\end{compactenum}
These assumptions are formalised as the following properties:
\begin{definition}[Knowledge Consistency]
  For all $i \in \ag$ and formulae $\phi \in \wkbl$,
  $\mbox{if }\kbi \der \phi \mbox{ then } \kbi \not \der \neg \phi.$
  \label{def:kb-consistent}
\end{definition}
\vspace{-5mm}
Enforcing knowledge consistency is straightforward.
Before adding any formula $\phi$ to $\kbi$ we check that $\kbi \cup \{\phi\} \not \der \neg \phi$.
\begin{definition}[Self-Awareness]
  For all $i \in \ag$ and formulae $\phi \in \wkbl$,
  $\mbox{if } \kbi \der \phi \mbox{ then } \kbi \der K_i \phi.$
  \label{def:self-awareness}
\end{definition}
\noindent \textbf{Remark 1.}
Self-awareness is not equivalent to the necessitation rule in \kdfouraxioma.
Necessitation states that if a $\phi$ is provable from no assumptions then $K_i \phi$ is provable from no assumptions as well~\cite{FHM+95rk}.
That is,
$\inference{\sat \phi}{\sat K_i\phi}$.
It requires $\phi$ to be a tautology.
On the other hand, self-awareness states that if $\phi$ is derivable from $i$'s knowledge, then $K_i \phi$ is also derivable.
For example, $\phi \vee \neg \phi$ is provable from no assumptions.
Therefore, from axiom A1 it is derivable $\kbi \der \phi \vee \neg \phi$ for all $\kbi$.
Consequently, by necessitation it also holds that $\kbi \der K_j \phi \vee \neg \phi$ for all $\kbi$ and $j \in \ag$.
However, consider now a predicate $p(\vv{t})$ which is not derivable from no assumptions.
It does not hold that $\kbi \der p(\vv{t})$ for all $\kbi$.
There is no axiom which includes $p(\vv{t})$ in the set of derivations of $\der$.
Nevertheless, self-awareness says that if $\kbi \der p(\vv{t})$ then $\kbi \der K_i p(\vv{t})$.
Note that, unlikely necessitation, we use the same agent $i$ in $\kbi$ and $K_i p(\vv{t})$.

\begin{example}
  \label{ex:1}
  Let $\snv$ be an SNM consisting of three agents Alice, Bob and Charlie, $\ag = \{\alice, $ $\bob, \charlie\}$; the friend request action, $ \actSet = \{\friendRequest\}$; and the connections Friend and Blocked, $\conSet = \{\Friend, \Blocked\}$.
  Here, we define $\dom$ to be a finite set of timestamps.

  Fig.~\ref{fig:snmexample} shows a graphical representation of \snv.
  In this model the dashed arrows represent connections.
  Note that the Friend connection is bidirectional, i.e., Alice is friend with Bob and vice versa.
  On the other hand, it is also possible to represent unidirectional connections, as Blocked; in \snv~Bob has blocked Charlie.
  Permissions are represented using a dotted arrow. In this example, Charlie is able to send a friend request to Alice.

  The predicates inside each node represent the agents' knowledge, e.g., Alice has $\post(\bob, \pub, 1)$ in her knowledge base, meaning that she knows that Bob posted at time 1 that he was in a pub.
  Similarly, Charlie's knowledge base contains the predicate $\post(\bob,\library,2)$ meaning that at time 2 Bob posted that he was in the library.
  Agents' nodes can also contain more complex \kbl\ formulae that may increase their knowledge.
  For instance, Alice knows $\loc(\bob,\pub,1)$ implicitly.
  Alice can in fact derive it by {\em Modus Ponens},
  from $\post(\alice, \pub, 1)$ and $\forall t. (\post(\alice, \pub, t) \implies \loc(\bob, \pub, t))$.
  The variable $t$ ranges over \dom, which, as mentioned earlier, consists in a finite set of timestamps.
  Being able to derive $\loc(\bob, \pub, 1)$ means that Alice knows that Bob's location at time 1 was a pub.
\end{example}

The satisfaction relation for \kbl~formulae, interpreted over SNMs, is defined as follows.

\begin{definition}[]
  Given an SNM $\snv = \langle \ag, \struct, \kb, \poval \rangle$,
  agents $i,j$ $\in \ag$, formulae $\phi, \psi \in \wkbl$, a finite
  set of agents $G \subseteq \ag$, $m \in \conSet$ and $n \in \actSet$,
  the {\em satisfaction relation} $\sat\ \subseteq \sn \times \kbl$
  is defined in Table \ref{tab:satkbl}.
\end{definition}

\begin{table}[!t]
  \centering
  \begin{tabular}{l c l}
    $\snv \sat p(\vv{t}) $& \mbox{iff} & $p(\vv{t}) \in \kb_{e}$\\
    $\snv \sat c_m(i,j)$ & \mbox{iff} & $(i,j) \in C_m$\\
    $\snv \sat a_n(i,j)$ & \mbox{iff} & $(i,j) \in A_n$\\
    $\snv \sat \neg \phi $ & \mbox{iff} & $\snv \not \sat \phi$ \\
    $\snv \sat \phi \wedge \psi $ & \mbox{iff} & $\snv \sat \phi \mbox{ and } \snv \sat \psi$\\
    $\snv \sat \forall x. \phi$ & \mbox{iff} & \mbox{ for all  } $v \in \dom$,  $\snv \sat \phi[v/x]$ \\
    $\snv \sat K_i \phi$ & \mbox{iff} & $\kbi \der \phi$\\
  \end{tabular}
  \caption{\kbl~satisfaction relation\label{tab:satkbl}}
\end{table}

The intuition behind the semantic definition of the knowledge modality is different in \kbl\ from that of epistemic logic.
As shown in Table~\ref{tab:kripkesat}, the accessibility relation in Kripke models captures the {\em uncertainty} of the agents.
It models all the states that an agent consider possible and knowledge is acquired when a given formula is true in all those states.
In SNMs, knowledge is explicitly present in the knowledge bases of the agents, hence modelling what the agents know rather than what they consider possible.
A given formula is known by an agent if it is present in her knowledge base or if she can derive it from her knowledge.
We use a special agent called {\em environment} (or simply $e$) which defines the truth of atomic formulae of the type $p(\vv{t})$.
The environment's knowledge base ($\kb_e$) contains all predicates which are true in the real world.
For instance, $\location(\alice, \mathit{Sweden})$ is in $\kb_e$ only if Alice's location is Sweden or, similarly, only if Bob's age is 20 the predicate $\age(\bob,20)$ is in $\kb_e$.

\begin{example}
  Let \snv~be the SNM in Fig.~\ref{fig:snmexample}.
  As described in Example~\ref{ex:1}, Alice knows that Bob posted that he was in a pub at time 1, meaning that $\snv \sat K_\alice \post(\bob, \pub, 1)$ holds.
  Indeed, it holds since $\post(\bob,\pub,1)$ is in the knowledge base of Alice, i.e., $\post(\bob,\pub,1) \in \kb_\alice$ and therefore it can be derived $\kb_\alice \der \post(\bob,\pub,1)$ (1).
Though not explicitly stated, it is possible for Alice to derive that Bob's location at time 1 was a pub, meaning that $\snv \sat K_\alice \loc(\bob,\pub,1)$ (2) should hold.
  Following the semantics of $K_i$ in Table~\ref{tab:satkbl}, the previous formula is true iff $\kb_\alice \der \loc(\bob,\pub,1)$.
  Fig.~\ref{fig:snmexample} shows that $\kb_\alice$ contains the formula $\forall t. (\post(\bob,\pub,t) \implies \loc(\bob,\pub,t))$ (3)---where $t$ is a timestamp
 ---therefore the deductive engine derives $\post(\bob,\pub,1) \implies \loc(\bob,\pub,1)$ (4).
From (1) and (4), by {\em modus ponens} we can derive $\loc(\bob,\pub,1)$, i.e., $\kb_\alice \der \loc(\bob,\pub,1)$, hence (2) holds.
\end{example}


\section{Model checking SNMs}\label{sec:mc}

In this section we present a model checking algorithm that directly implements the semantics of \kbl~in Table~\ref{tab:satkbl}, and we show that model checking is decidable under the following assumptions:
\begin{assumption}
  \label{asp:finitedomains}
  All domains are finite.
\end{assumption}
\begin{assumption}
  \label{asp:computablefunctions}
  All functions are computable.
\end{assumption}

These assumptions are present in all real social networks. Domains in SNMs might be, the set of users, posts, pictures, likes, tags and so on. In practice at any moment in time there is a finite amount of any of these elements. Consequently, when having a universal quantification over a domain it is reasonable to consider only the finite set of elements in the domain at that concrete moment in time. Furthermore, we assume that functions in \kbl~terms must be computable. As mentioned in the introduction, \kbl~is a logic embedded in a framework to express privacy policies. The framework includes the notion of instantiation where all the elements of SNMs are instantiated for a concrete social network.
For instance, in~\cite{PS14fpp} 
we presented the instantiations of Facebook and Twitter. In these instantiations functions were used to retreive information, e.g., $\mathit{followers}(u)$ which returns all the followers of the user or $\mathit{friends}(u)$ which returns all the friends of $u$. Another type of functions could be $\mathit{weather}(\mathit{London})$ or $\mathit{location}(u)$, which return the current weather in London and $u$'s current location, respectively. Therefore, computable functions are enough for the practical use of the logic.

\begin{theorem}
  Let \snv~be an SNM and $\phi \in \wkbl$ be a formula.
  Determining whether $\snv \sat \phi$ is decidable.
  \label{thm:mcdecidable}
\end{theorem}

\begin{proof}
  We show decidability of the model checking problem for \kbl\ by presenting an algorithm which implements the semantics of Table~\ref{tab:satkbl},

  First, we expand the universal quantifiers in $\phi$ by inductively transforming each subformula $\forall x.\phi'$ into a conjunction with one conjunct $\phi'[v/x]$ for each element $v$ of the domain \dom. Given that the domain is finite (see Assumption~\ref{asp:finitedomains}), it always terminates and results in a quantifier free formula.
  Secondly, we compute all functions and replace all constants with an element of the domain according to the assignment in \struct. From Assumption~\ref{asp:computablefunctions}, we can deduce that this step always terminates. After this step we are left with a quantifier free formula without functions or constant symbols.
  Finally, we inductively show that all the elements of the formula (see Def.~\ref{kbl:syntax}) can be computed.
    \begin{compactitem}
    \item Checking $c_m(i,j)$ and $a_n(i,j)$ can be performed in constant time, simply by checking $(i,j) \in C_m$ or $(i,j) \in A_n$, respectively.
    \item Checking $p(\vv{t})$ requires the query $p(\vv{t}) \in \kb_e$ to the environment's knowledge base.
          It can be performed in constant time.
    \item $\neg \phi$ and $\phi_1 \wedge \phi_2$ can be done in constant time, using the induction hypothesis.
    \item $K_i \phi$ requires a query to the epistemic engine to determine  $\kbi \der \phi$.
          Solving the previous query is a decidable problem \cite{FHM+95rk}.
    \end{compactitem}
    The algorithm goes recursively from the top most element of $\phi$ to the bottom.
\end{proof}

In Section~\ref{sec:complexity} we study the complexity of this algorithm and compare it to that of model checking in traditional Kripke models.
Nevertheless, in order to provide a fair comparison, we first show that the same set of properties of knowledge that are sound w.r.t.~Kripke models are also sound w.r.t.~SNMs.

\section{Properties of Knowledge in SNMs}
\label{sec:axioms}

Here we explore properties of knowledge in SNMs.
 In particular, we consider the axioms of some of the standard axiomatisations for epistemic logic, and prove that such axioms are sound with respect to SNMs.

In \cite{FHM+95rk} Fagin~\etal show which properties of knowledge are sound w.r.t. Kripke models depending on the type of accessibility relation of the model.
For instance, the following axiom is sound w.r.t. the set of Kripke models where the accessibility relation is reflexive: $\mbox{(A3) } K_i \phi \implies \phi.$

These properties of knowledge comprise the different axiomatisations of epistemic logic.
In SNMs the properties of knowledge will depend on the axiomatisation from epistemic logic~\cite{FHM+95rk} that we choose for $\der$.
As we described in Def.~\ref{def:derivation}, \der\ includes all the axioms and derivation rules from \kdfouraxioma.


In epistemic logic one can talk about {\em knowledge} or {\em belief} depending on the properties (or axiomatisations) that are sound w.r.t.~a particular set of Kripke models.
Axiom A3 is commonly called {\em Knowledge axiom}.
It means that the facts agents know are true.
When this axiom is not present, the ``knowledge'' of the agents is regarded as belief.
As you might have noticed, in SNMs the truth of the facts that the agents know is not linked to whether they are true or not.
For example, imagine that Alice knows that Bob and Charlie are friends, i.e., $K_\alice \friend(\bob,\charlie)$, which is true iff $\kb_\alice \der \friend(\bob, \charlie)$.
This is not connected to the actual truth of the predicate $\friend(\bob, \charlie)$, which holds iff $(\bob, \charlie) \in C_\Friend$. When the knowledge axiom is not present, some philosophers argue that it is required that the beliefs of the agents are consistent. This is captured by the following axiom, where $\bot$ represents \emph{falsum}: $\mbox{(D) } \neg K_i \bot.$

In Kripke models, axiom D is present when the accessibility relation is serial~\cite{FHM+95rk}. In SNMs, we assume agents' knowledge bases to be consistent (see Def.~\ref{def:kb-consistent}). Therefore, $\bot$ cannot be derived.
\begin{lemma}
  \label{lemma:Dholdspaper}
  Axiom D is sound with respect to SNMs.
\end{lemma}
As we mentioned in the introduction, \kbl~and SNMs were developed in the context of a privacy policy framework for social networks~\cite{PS14fpp,PardoLic}.
In privacy policies it is more natural to write ``Alice cannot know my location'' than ``Alice cannot belief my location''. Because of this, we chose to talk about knowledge, even though we are dealing with an axiomatisation for belief.

The most basic set of properties for Kripke models, i.e., the set of properties that are sound w.r.t.~Kripke models with no conditions in their accessibility relation, is the \kaxioma~axiomatisation \cite{FHM+95rk}.
It consists of two axioms and two inference rules.
Given $\phi \in \lang$ and $i \in \ag$,
\begin{compactitem}[]
  \item A1. All (instances of) first-order tautologies,
  \item A2. $(K_i\phi \wedge K_{i}(\phi \implies \psi)) \implies K_i\psi$, 
  \item R1. From $\phi$ and $\phi \implies \psi$ infer $\psi$, 
  \item R2. From $\phi$ infer $K_i \phi$ where $\phi$ must be provable from no assumptions. 
\end{compactitem}

\begin{lemma}
  \label{lemma:ksoundnesspaper}
  \kaxioma~is sound with respect to SNMs.
\end{lemma}

The axioms and inferences rules of \kaxioma, together with axiom D comprises the axiom system \kdaxioma.
Nevertheless, there exist two more axioms that are normally present in knowledge and belief axiomatisations, the so called {\em positive introspection} (A4) and {\em negative introspection} (A5)~\cite{FHM+95rk}.
The former expresses that agents in the system are aware of their knowledge, the latter means that agents know everything that they do not know.
Given $\phi \in \lang$ and $i \in \ag$

\begin{compactitem}[]
\item A4. $K_i \phi \implies K_i K_i \phi$,
\item A5. $\neg K_i \phi \implies K_i \neg K_i \phi$.
\end{compactitem}
\begin{lemma}
  \label{lemma:a4a5paper}
  Axiom A4 is sound with respect to SNMs.
\end{lemma}
\begin{lemma}
  Axiom A5 is not sound with respect to SNMs.
\end{lemma}
A4 follows from our assumption that agents are self-aware of their knowledge (see Def.\ref{def:self-awareness}).
On the other hand, A5 does not follow given the current set of assumptions in knowledge bases.
An agent's knowledge base does not contain any knowledge regarding what she does not know, unless it is explicitly inserted.

The axiomatisation \kaxioma~together with axioms D and A4 forms the so-called \kdfouraxioma~axiomatisation.
We thus have the following result for SNMs.

\begin{theorem}
  \label{thm:kd45holdspaper}
  \kdfouraxioma\ is sound with respect to SNMs.
\end{theorem}

\subsubsection*{Common Knowledge}
\label{sec:common}

Here we introduce the notion of {\em common knowledge}, which we
represent using the modality $C_G$ where $G$ is a group of agents. A
fact becomes common knowledge when everybody knows it, and also,
everyone knows that everyone knows it, and so forth. This is a useful
concept in the social network setting. Consider the effect of publishing a post
$p(\vv{t})$ in a social network. After posting, the owner of the post and the
audience will know the post, $E_{\{\owner\} \cup \audience} ~
p(\vv{t})$. Moreover, the owner also will know that everyone who was
included in the audience will know the post, $K_{\owner}E_{\audience}
~ p(\vv{t})$. But even more, each of the users in the audience will
know that each other knows the post, i.e. $E_{\{\owner\} \cup
  \audience} E_{\{\owner\} \cup \audience} ~ p(\vv{t})$ and so on. The
traditional definition of common knowledge \cite{FHM+95rk} over Kripke
models accurately captures the described effect. Given a Kripke model
$M$, a state $s \in M$, a formula $\phi \in \lang$ and a set of agents
$G$, common knowledge is defined as follows:
$ (M, s) \sat C_G \phi \mbox{ iff } (M, s) \sat E^{k}_G \phi \mbox{
  for } k = 1 \ldots $ where $E^{0}_G \phi = \phi$ and $E^{k+1}_G
\phi = E_G \phi E^{k}_G \phi$. The definition of common knowledge for
SNMs is analogous to the one above.
\begin{definition}[]
  Given an SNM $\snv$, a formula $\phi \in \wkbl$ and a set of agents
  $G$, common knowledge is defined as follows:
  $ \snv \sat C_G \phi \mbox{ iff } \snv \sat E^{k}_G \phi \mbox{ for
  } k = 1 \ldots $
\end{definition}

Given formulae $\phi, \psi \in \lang$, the set $G \subseteq \ag$ and $i \in \ag$, the following axiomatisation characterises common knowledge \cite{FHM+95rk}:
\begin{compactitem}[]
\item C1. $E_G \phi \Longleftrightarrow \bigwedge_{i \in G} K_i \phi$,
\item C2. $C_G \phi \Longleftrightarrow E_G (\phi \wedge C_G \phi)$,
\item RC1. From $\phi \implies E_G (\psi \wedge \phi)$ infer $\phi
  \implies C_G \psi$ where $\phi \implies E_G (\psi \wedge \phi)$ must
  be provable from no assumptions.
\end{compactitem}
\noindent
\begin{lemma}
  \label{lemma:commonholdspaper}
  The axioms C1 and C2, and the rule RC1 are sound w.r.t.~SNMs.
\end{lemma}

\subsubsection*{Distributed Knowledge}
\label{sec:distributed}

In this section we introduce the distributed knowledge operator,
represented by the modality $D_G$. A fact becomes distributed
knowledge in the group of agents $G$ when it is known by combining the
knowledge of all individual agents. It can be seen as a wise agent. In
Kripke models, distributed knowledge is defined by removing possible
states, i.e., removing uncertainty. Formally,
$(M, s) \sat D_G \phi \mbox{ iff } (M, t) \sat \phi \mbox{ for all $t$ such that}~(s,t) \in
\bigcap_{i \in G}\mathcal{K}_i.$
We define distributed
knowledge as the union of all the explicit knowledge that all the
agents in $G$ have and everything that can be derived from it.


\begin{definition}[Distributed knowledge]
    Given an SNM $\snv$, a formula $\phi \in \wkbl$ and a set of
    agents $G$, distributed knowledge is defined as follows:
    $\snv \sat D_G \phi \mbox{ iff }  \bigcup_{i \in G} \kb_i \der \phi.$
\end{definition}


The following axioms characterise distributed knowledge \cite{FHM+95rk}:

\begin{compactitem}[]
\item D1. $D_{\{i\}} \phi \Longleftrightarrow K_i \phi,~i = 1, \ldots, n$,
\item D2. $D_G \phi \implies D_{G'} (\phi)$ if $G \subseteq G'$,
\item DA2 and DA4. Axioms A2 and A4 of \kdfouraxioma, $K_i$ with $D_G$ in each axiom.
\end{compactitem}

Note that axiom D is not required because we work with a belief axiomatisation~\cite{FHM+95rk}.
Therefore, it is possible for a group of agents to have inconsistent distributed beliefs.
In what follows, we show that this axiomatisation for Kripke models is sound with respect to SNMs as well.

\begin{lemma}
  \label{lemma:distributedholdspaper}
  Axioms D1 and D2, together with the axioms A2 and A4 of the
  \kdfouraxioma-axiomatisation (replacing the modality $K_i$ with
  the modality $D_G$) are sound w.r.t.~SNMs.
\end{lemma}

\section{Translation of SNMs into Kripke Models}
\label{sec:canonical}

In this section, we show that SNMs can be encoded into Kripke models. Our proof is constructive, starting from an SNM we give a procedure to build a {\it canonical} Kripke model, and we prove that satisfaction is preserved when interpreting \kbl~formulae as epistemic logic formulae.



For epistemic logic, Fagin \etal show that it is possible to construct
a canonical Kripke model which satisfies a given formula
$\phi$~\cite{FHM+95rk}, provided that $\phi$ is consistent with
respect to some of the axiomatisations of knowledge. A formula $\phi$
is \emph{\kdfouraxioma-consistent} if $\neg \phi$ cannot be
derived. A set of formulae is \kdfouraxioma-consistent if the
conjunction of all the formulae in the set is
\kdfouraxioma-consistent. We say that a set of formulae $\Phi$ is
\emph{maximal \kdfouraxioma-consistent} with respect to the
language \lang, if $\Phi$ is \kdfouraxioma-consistent and for all
$\phi$ in \lang~but not in $\Phi$, the set $\Phi \cup \{\phi\}$ is not
\kdfouraxioma-consistent. In what follows, we describe the
procedure of how to construct a canonical Kripke model for a
\kdfouraxioma-consistent formula.  We will follow a similar
approach when translating SNMs into Kripke models.

\begin{definition}[Canonical Kripke model for \kdfouraxioma \cite{FHM+95rk}]
  Consider a {\em \kdfouraxioma-consistent} formula $\phi$.
  Let $\sub(\phi)$ be the set of all subformulae of $\phi$.
  We define $\subplus(\phi)$ to be the set of all subformulae and their negations, i.e. $\subplus(\phi) = \sub(\phi) \cup \{ \neg \psi ~ | ~ \psi \in \sub(\phi) \}$.
  We also define $\con(\phi)$ to be the set of maximal \kdfouraxioma-consistent subsets of
  $\subplus(\phi)$
  . Given a set of formulae $\Theta \subseteq \lang$, we define $\Theta / K_i = \{\phi~|~K_i\phi \in
  \Theta\}$. The canonical Kripke model for $\phi$ is defined as follows:
  $ M_{\phi} = \langle S_\phi, \pi, \{\mathcal{K}_i\}_{i \in Ag}
  \rangle$ where
  $S_\phi = \{ s_{\Theta}~|~\Theta \in \con(\phi)  \}$,
  $\accrelation_i = \{ (s_{\Theta}, s_{\Psi})~|~\Theta / K_i \subseteq \Psi / K_i,~\Theta / K_i \subseteq \Psi \}$ and

  \noindent$\pi(s_{\Theta})(p(t_1,\ldots,t_k)) =
    \begin{cases}
      ~\textbf{true} & \mbox{if }~p(t_1,\ldots,t_k) \in \Theta\\
      ~\textbf{false} & \mbox{if }~p(t_1,\ldots,t_k) \not \in \Theta
    \end{cases}
    $
  \label{def:canonical}
\end{definition}

Fagin \etal show that $\phi$ is satisfiable in the resulting canonical Kripke model~\cite[Theorem 3.2.4]{FHM+95rk}.
The set of Kripke models that are sound and complete with respect to \kdfouraxioma~are the ones with a serial and transitive accessibility relation.
The accessibility relation of the previous canonical Kripke model is, as shown in~\cite[Theorem 3.2.4]{FHM+95rk}, serial and transitive. We denote the set of Kripke models with the previous type of accessibility relation as \mlt.

The canonical Kripke model will have at most $2^{|\phi|}$ states, as shown in~\cite[Theorem 3.2.4]{FHM+95rk} where $|\phi|$ is the length of the formula $\phi$. Even though it is finite, this approach of constructing a Kripke model can lead to an exponential growth of the size of the model. For example, if we assume that the knowledge of the agents increases monotonically, i.e., agents do not forget any knowledge they have previously obtained, then the size of $\phi$ will have a lower bound, from which its size will only grow, and consequently, the size of the corresponding canonical Kripke model. In what follows, we define a function which takes an SNM and converts it into the corresponding canonical Kripke model.

First we describe how to construct a set containing all the true formulae in an SNM, called the characteristic set of the social network.

\begin{definition}[]
  \label{def:characteristicset}
  The \emph{characteristic set} of an SNM \snv, denoted as
  $\Phisnv$, is constructed as follows: $\Phisnv = \{p(\vv{t}) ~ | ~ p(\vv{t}) \in \kb_e \} ~ \cup ~ \{K_i \phi ~ | ~ \phi \in \kb_i \} $ $~ \cup ~ $ $\{c(i,j) ~ | ~ (i,j) \in C_c, c \in \conSet \} ~ $ $\cup ~ \{a(i,j) ~ | ~ (i,j) \in A_a, a \in \actSet \}.$
\end{definition}

\noindent
Moreover, we define the \emph{characteristic formula} of an SNM.

\begin{definition}[]
  \label{def:characteristicformula}
  Given a characteristic set, \Phisnv, of an SNM \snv, its
  \emph{characteristic formula}, denoted as $\phisnv$, is defined as
  $\phisnv = \bigwedge_{\psi \in \Phisnv} \psi$.
\end{definition}
We will use the characteristic formula of an SNM to create the
corresponding Kripke model, therefore we must show that this formula
is \kdfouraxioma-consistent.

\begin{lemma}
  \label{lemma:phisnvconsistentpaper}
  For all $\snv \in \sn$, \phisnv~is \kdfouraxioma-consistent.
\end{lemma}


\noindent
We are now ready to provide our translation from SNMs into
canonical Kripke models.

\begin{definition}[Kripke transformation function]
Let $\kt: \sn \rightarrow \mlt$ be a function which takes an SNM and
converts it to the corresponding Kripke model as follows. Given an
$\snv \in \sn$, $\kt(\snv)$ is defined as follows:
\begin{inparaenum}[1)]
\item Construct $\Phisnv$ as defined in Def.~\ref{def:characteristicset};
\item Construct $\phisnv$ as defined in Def.~\ref{def:characteristicformula};
\item Return the resulting canonical Kripke model of $\phisnv$ as defined in Def.~\ref{def:canonical}.
\end{inparaenum}
\label{def:ktfunction}
\end{definition}

\noindent
We thus have our main theorem.

\begin{theorem}
  \label{thm:sntokripkepaper}
  If a formula $\phi$ is satisfied in an SNM $\snv$ then $\phi$ is
  satisfied in the Kripke model $\kt(\snv)$.
\end{theorem}


\subsection{Translation of Kripke Models into SNMs}
\label{sec:kripketosnm}



Note that, in general, it is not possible to translate arbitrary Kripke models into SNMs. One of the reasons is that in Kripke models there exists only one type of predicate, which is always interpreted in the same way, whereas in SNMs, there are three types of predicates. We cannot even translate back canonical Kripke models constructed using \kt. To see why, let us consider a canonical Kripke model with the following characteristic set of formulae $\{K_\alice \loc(\bob,\library), \friend(\alice,\bob)\}$. We know that the predicate $\loc(\bob,\library)$ belongs to Alice's knowledge base, since it is under the scope of a knowledge modality. However, we cannot know the type of the predicate $\friend(\alice,\bob)$, it could be part of a connection relation, action relation or simply be a regular predicate which should appear in the environment's knowledge base.

That said, we show here that it is in fact always possible to reconstruct the original SNM from the canonical Kripke model, if we slightly modify our translation function \kt.
Let \mPhisnv~be a \emph{marked characteristic set}, which is a characteristic set as defined in Def.~\ref{def:characteristicset}, but having the predicates annotated so that their type can be syntactically identified. For example, if the predicate above $\friend(\alice,\bob)$ is a connection predicate, it would be converted to $\co\_\friend(\alice,\bob)$. We can now define \ktm~to be a Kripke transformation function as in Def.~\ref{def:ktfunction}, except for the input characteristic set, which is replaced by \mPhisnv.
Given that we can uniquely identify the type of the predicates it is trivial to define a function that takes a Kripke model constructed using \ktm~and returns the equivalent SNM. The function proceeds as follows: firstly, it searches for all the agents present in all formulae and subformulae in $\mPhisnv$ and creates one node per agent; secondly, it puts regular predicates in the environment's knowledge base; thirdly it creates relations between agents for each connection and permission predicate; finally, for all formulae of the form $K_i \phi$ it includes $\phi$ in $i$'s knowledge base.
We refer the reader to the extended version of this paper for the formal definitions of \mPhisnv, \ktm~and the SNM construction.

We also show that satisfaction is preserved between a marked canonical Kripke model and its original SNM when formulae are evaluated in the state corresponding to the marked characteristic set ($s_{\mPhisnv}$).

\begin{theorem}
  \label{thm:kripketosnmpaper}
  If a formula $\phi$ is satisfied in the state $s_{\mPhisnv}$ of a
  Kripke model $\ktm(\snv)$ then $\phi$ is satisfied in the SNM
  $\snv$.
\end{theorem}


\newcommand\formulasize[1]{|#1|}
\newcommand\modelsize[1]{||#1||}

\section{Model checking complexity}
\label{sec:complexity}

In~\cite{FHM+95rk}, Fagin \etal 
prove that the complexity of the model checking problem for \kdfouraxioma~(without common and distributed knowledge) is PSPACE-complete for $n$ agents where $n > 1$ and NP-complete for one agent. They also prove that for a model $M = (S, \pi, \mathcal{K}_1, \ldots, \mathcal{K}_n)$ \emph{``There is an algorithm that, given a structure $M$, a state $s$ of $M$ and a formula $\phi \in \lang$, determines, in time $O(\modelsize{M} \times \formulasize{\phi})$, whether $(M, s) \sat \phi$''} (see \cite[Proposition 3.2.1]{FHM+95rk}) where $\modelsize{M}$ is the sum of all the states in $S$ and the number of pairs in all $\mathcal{K}_i$, and $\formulasize{\phi}$ is the length of the formula defined as usual. This algorithm is not optimal, but the result is useful to compare the model checking problem in SNMs and the Kripke models constructed using our translation. 


Let $M_{\phisnv}$ be the model $\kt(\snv)$ for an SNM \snv. The complexity of the model checking problem of a formula $\phi$ in the previous model is
$  O(\modelsize{M_{\phisnv}} \times \formulasize{\phi}).$
$M_{\phisnv}$ has size at most $2^{\formulasize{\phisnv}}$ (see Section~\ref{sec:canonical}), therefore it holds $\modelsize{M_{\phisnv}} \leq 2^{\formulasize{\phisnv}}$.
Thus, for simplicity and w.l.o.g.~the above may be rewritten as
$  O(2^{\formulasize{\phisnv}} \times \formulasize{\phi}).$

In what follows we study the complexity of the model checking problem in \kbl. The proof of Theorem~\ref{thm:mcdecidable} describes an algorithm to determine whether $\snv \sat \phi$. We consider \kbl~without common and distributed knowledge, since the complexity for Kripke models mentioned at the beginning of the section also excludes these modalities. For simplicity in the complexity analysis and w.l.o.g.~we only consider quantifier free formulae which do not contain functions.

Let $M_{\kbi}$ be the canonical Kripke model resulting from the conjunction of all formulae in agent's $i$ knowledge base using our translation, the complexity of the model checking problem is given by the function \emph{checking complexity} ($\cc$): $\cc(p(\vv{t})) = \cc(c(i,j)) = \cc(a(i,j) = c$, $\cc(\neg \phi) = 1 + \cc(\phi)$, $\cc(\phi_1 \wedge \phi_2) = 1 + \cc(\phi_1) + \cc(\phi_2)$ and $\cc(K_i \phi) = O(\modelsize{M_{\kbi}} \times \formulasize{\phi})$
%
where $c$ is an upper-bound in the cost of checking satisfaction of predicates in the environment's knowledge base, connection predicates and action predicates. Negation and conjunction need one step plus the complexity of checking satisfaction of their subformulae. Finally, satisfaction of $K_i \phi$ depends on checking $\kbi \der \phi$, which requires solving the model checking problem as defined for Kripke models. Therefore it has the same complexity.
Let $\outerk : \wkbl \rightarrow 2^{\wkbl}$ be a function that takes a \kbl~formula and returns the set of subformulae where $K_i$ is the top most operator and it is not under the scope of a knowledge modality. For example, $\outerk(K_a(p(s) \wedge K_b q(s)) \wedge p(u) \wedge \neg K_b r(s) \wedge K_c u(v) ) = \{K_a(p(s) \wedge K_b q(s)), K_b r(s), K_c u(v)\}$. Note that $K_b q(s)$ is not part of the set because it is under the scope of $K_a$. The complexity of checking whether a formula $\phi$ is satisfiable in an SNM is $O(\sum_{K_i \phi_i \in \outerk(\phi)} (\modelsize{M_{\kbi}} \times \formulasize{\phi_i}) + m_\phi)$
where $m_\phi \in \nat$. The characteristic formula of an agent's knowledge base is the conjunction of all its knowledge, which we denote as $\phi_{\kbi}$. As before, it holds that $\modelsize{M_{\kbi}} < \formulasize{2^{\phikbi}}$, which we use again for the complexity of the problem $O(\sum_{K_i \phi_i \in \outerk(\phi)} (2^{\formulasize{\phikbi}} \times \formulasize{\phi_i}) + m_\phi)$.

The intuition is as follows: $m_\phi$ is the cost of checking predicates, conjunctions and negations in $\phi$, which we assume to be some constant that depends on the length of $\phi$. Besides, $\sum_{K_i \phi_i \in \outerk(\phi)} $ $(2^{\formulasize{\phikbi}} \times \formulasize{\phi_i})$ is the cost of checking each subformula $\phi_i$ in the knowledge base of the corresponding agent. In short, we have replaced checking satisfaction of $\phi$ in a complete model of the social network to checking satisfaction of subformulae of $\phi$ in the corresponding knowledge bases of the agents.

Checking the parts of $\phi$ that only contain predicates and logical connectives has very similar complexity in both models. In the canonical Kripke model of an SNM \snv, the state corresponding to the characteristic set ($s_{\Phisnv}$) contains all true predicates (see Def.~\ref{def:canonical}). Similarly, in SNMs it is only needed to check the environment's knowledge base, and the connection and action relations (see Table~\ref{tab:satkbl}). In both cases the complexity is determined by the length of this particular part of $\phi$. Therefore, in order to compare the complexity of the model checking problem, we only focus on the parts of the formula that are under the scope of a knowledge modality. Given a formula $\phi$, let $\phi^K$ be the conjunction of the subformulae starting with a $K_i$ modality (for any $i \in \ag$), formally, $\phi^K \triangleq \bigwedge_{\psi \in \outerk(\phi)}\psi$. Thus the complexity of the model checking problem in Kripke models is reduced to $O(2^{\formulasize{\phisnv}} \times \formulasize{\phi^K})$, and in SNMs it is $O(\sum_{K_i \phi_i \in \outerk(\phi)} (2^{\formulasize{\phikbi}} \times \formulasize{\phi_i}))$.
To formally compare the complexity of the problem in both models we prove the following.

\begin{lemma}
  Given $\snv \in \sn$ and a formula $\phi$ the following holds:
$O(\sum_{K_i \phi_i \in \outerk(\phi)} (2^{\formulasize{\phikbi}} \times \formulasize{\phi_i})) < O(2^{\formulasize{\phisnv}} \times \formulasize{\phi^K}).$
  \label{lemma:complexity}
\end{lemma}

The previous lemma shows that it is always more efficient to check satisfaction of a formula $\phi$ in SNMs. Intuitively, it shows that it is more efficient to construct Kripke models representing the agents' knowledge base and locally check the corresponding subformulae, than constructing the complete Kripke model to check the conjunction of the mentioned subformulae. The difference in complexity becomes more apparent as less agents are involved in the knowledge modalities of $\phi$. When an agent is not mentioned in $\phi$ her knowledge base is disregarded. For instance, in the SNM of Fig.~\ref{fig:snmexample} checking $K_\charlie \loc(\bob,\pub,1)$ requires (at most) $2^4 + 5 = 21$ steps where $4$ is the size of the formula in \charlie's knowledge base and $5$ is the size of $K_\charlie \loc(\bob,\pub,1)$, whereas in the corresponding canonical Kripke model it requires (at most) $2^{4+14+12}+5 = 1073741829$ steps where $14$ is the size of the conjunction of all the formulae in the knowledge base of $\alice$ (assuming that the domain of $x$ only has one element), and $12$ is the size of the predicates $friend(\alice,\bob)$, $friend(\bob,\alice)$, $\blocked(\bob,\charlie)$ and $\friendRequest(\charlie,\alice)$.


\section{Related work}
\label{sec:related}

The use epistemic logic to model knowledge in social networks is not new. One line of work consists in using two dimensional modal logic. It relies on Kripke models where the knowledge of the agents in the social network is encoded using an accessibility relation, and friendship is represented using a symmetric irreflexive relation between agents \cite{SLG13fel}. Other epistemic logics include a public (and private) announcement operator to study diffusion of information in the network \cite{RT11lkfsn, Christoff201548}. Permission and knowledge has also been merged in the so called deontic-epistemic logic \cite{ABT11dlpc}. For a detailed comparison among these logics and \kbl~we refer to the work by Pardo \& Schneider \cite{PS14fpp, PardoLic} and references therein.

There exist several model checkers for epistemic logic that perform efficiently in rather large scenarios \cite{GvdM04MCK,van2007demo,lomuscio2015mcmas}. However, as shown in this paper, model checking in the canonical Kripke model constructed from an SNM has higher complexity than in the SNM.

On the other hand, the model checking algorithm presented in this paper requires checking whether $\kbi \der \phi$. As mentioned in Section~\ref{sec:ppf}, this check can be resolved by using any of the existing model checkers or SAT solvers for epistemic logic. For this reason, any improvement in the efficiency of the model checking problem in Kripke models, will also be improve the performance when checking formulae in the individual knowledge bases of each agent. 
In addition, local checks in different knowledge bases can easily be parallelised. For instance, if there is one process per knowledge base, formulae regarding different agents' knowledge can be checked in parallel in the corresponding knowledge bases. To the best of our knowledge, there are no parallel model checkers for epistemic logic.


\section{Final Discussion}
\label{sec:conclusion}

We have proved that the model checking problem in SNMs is decidable.
We have shown the relation between SNMs and Kripke models.
Concretely, we have proven that the belief axiomatisation \kdfouraxioma, which was originally defined for epistemic logic and naturally models agents' reasoning, is sound w.r.t.~SNMs.
We have provided a translation of SNMs models into canonical Kripke models and proved that satisfaction of any formula in the SNM is preserved in the corresponding Kripke model.
We have also provided a translation from the canonical Kripke structure (obtained from our translation from SNMs) into the original SNM.
We have proven that all formulae are satisfied in the state corresponding to the characteristic set of the SNM in the Kripke model are also satisfied in the original SNM.
Finally, we showed the model checking problem in SNMs using our algorithm is more efficient than using the standard Kripke semantics.

We conjecture that arbitrary Kripke models (in the frame of models with serial and transitive relations) can be translated to SNMs.
However, to preserve satisfaction the translation would generate several SNMs from a given Kripke model.
Each of these SNMs would correspond to a state in the Kripke model.

The semantics of the privacy policy language \ppl~(included in \ppf) is given in terms of the satisfaction relation of \kbl, so \ppl~conformance is reduced to \kbl~satisfaction.
Thanks to our results we may check conformance of  \ppl~policies by using existing model checkers for epistemic logic.

\subsubsection*{Acknowledgements}
This research has been supported by: the Swedish funding agency SSF under
the grant {\em Data Driven Secure Business Intelligence} and the Swedish Research Council ({\it Vetenskapsr\aa det}) under grant Nr.~2015-04154 ({\em PolUser: Rich User-Controlled Privacy Policies}).

\bibliographystyle{eptcs}
\bibliography{references}

\begin{thebibliography}{10}
\providecommand{\bibitemdeclare}[2]{}
\providecommand{\surnamestart}{}
\providecommand{\surnameend}{}
\providecommand{\urlprefix}{Available at }
\providecommand{\url}[1]{\texttt{#1}}
\providecommand{\href}[2]{\texttt{#2}}
\providecommand{\urlalt}[2]{\href{#1}{#2}}
\providecommand{\doi}[1]{doi:\urlalt{http://dx.doi.org/#1}{#1}}
\providecommand{\bibinfo}[2]{#2}

\bibitemdeclare{article}{ABT11dlpc}
\bibitem{ABT11dlpc}
\bibinfo{author}{Guillaume \surnamestart Aucher\surnameend},
  \bibinfo{author}{Guido \surnamestart Boella\surnameend} \&
  \bibinfo{author}{Leendert \surnamestart van~der Torre\surnameend}
  (\bibinfo{year}{2011}): \emph{\bibinfo{title}{A dynamic logic for privacy
  compliance}}.
\newblock {\sl \bibinfo{journal}{Artificial Intelligence and Law}}
  \bibinfo{volume}{19}(\bibinfo{number}{2-3}), pp. \bibinfo{pages}{187--231},
  \doi{10.1007/s10506-011-9114-3}.

\bibitemdeclare{article}{Christoff201548}
\bibitem{Christoff201548}
\bibinfo{author}{Zo\'e \surnamestart Christoff\surnameend} \&
  \bibinfo{author}{Jens~Ulrik \surnamestart Hansen\surnameend}
  (\bibinfo{year}{2015}): \emph{\bibinfo{title}{A logic for diffusion in social
  networks}}.
\newblock {\sl \bibinfo{journal}{Journal of Applied Logic}}
  \bibinfo{volume}{13}, pp. \bibinfo{pages}{48 -- 77},
  \doi{10.1016/j.jal.2014.11.011}.

\bibitemdeclare{techreport}{van2007demo}
\bibitem{van2007demo}
\bibinfo{author}{Jan \surnamestart van Eijck\surnameend}
  (\bibinfo{year}{2007}): \emph{\bibinfo{title}{DEMO -- A Demo of Epistemic
  Modelling}}.
\newblock \bibinfo{type}{Technical Report}, \bibinfo{institution}{Amsterdam
  University Press}.

\bibitemdeclare{book}{K14cn}
\bibitem{K14cn}
\bibinfo{author}{Kayhan \surnamestart Erciyes\surnameend}
  (\bibinfo{year}{2014}): \emph{\bibinfo{title}{Complex Networks: An
  Algorithmic Perspective}}, \bibinfo{edition}{1st} edition.
\newblock \bibinfo{publisher}{CRC Press, Inc.}, \bibinfo{address}{Boca Raton,
  FL, USA}, \doi{10.1201/b17409}.

\bibitemdeclare{book}{FHM+95rk}
\bibitem{FHM+95rk}
\bibinfo{author}{Ronald \surnamestart Fagin\surnameend},
  \bibinfo{author}{Joseph~Y \surnamestart Halpern\surnameend},
  \bibinfo{author}{Yoram \surnamestart Moses\surnameend} \&
  \bibinfo{author}{Moshe~Y \surnamestart Vardi\surnameend}
  (\bibinfo{year}{2003}): \emph{\bibinfo{title}{Reasoning about Knowledge}}.
\newblock \bibinfo{publisher}{The MIT press}, \bibinfo{address}{Cambridge, MA,
  USA}.

\bibitemdeclare{incollection}{GvdM04MCK}
\bibitem{GvdM04MCK}
\bibinfo{author}{Peter \surnamestart Gammie\surnameend} \& \bibinfo{author}{Ron
  \surnamestart van~der Meyden\surnameend} (\bibinfo{year}{2004}):
  \emph{\bibinfo{title}{MCK: Model Checking the Logic of Knowledge}}.
\newblock In: {\sl \bibinfo{booktitle}{CAV}}, {\sl \bibinfo{series}{LNCS}}
  \bibinfo{volume}{3114}, \bibinfo{publisher}{Springer}, pp.
  \bibinfo{pages}{479--483}, \doi{10.1007/978-3-540-27813-9_41}.

\bibitemdeclare{inproceedings}{HC13bsal}
\bibitem{HC13bsal}
\bibinfo{author}{Andrew~K. \surnamestart Hirsch\surnameend} \&
  \bibinfo{author}{Michael~R. \surnamestart Clarkson\surnameend}
  (\bibinfo{year}{2013}): \emph{\bibinfo{title}{Belief Semantics of
  Authorization Logic}}.
\newblock In: {\sl \bibinfo{booktitle}{CCS}}, \bibinfo{publisher}{ACM}, pp.
  \bibinfo{pages}{561--572}, \doi{10.1145/2508859.2516667}.

\bibitemdeclare{inproceedings}{YKBA11afps+}
\bibitem{YKBA11afps+}
\bibinfo{author}{Yabing \surnamestart Liu\surnameend},
  \bibinfo{author}{Krishna~P. \surnamestart Gummadi\surnameend},
  \bibinfo{author}{Balachander \surnamestart Krishnamurthy\surnameend} \&
  \bibinfo{author}{Alan \surnamestart Mislove\surnameend}
  (\bibinfo{year}{2011}): \emph{\bibinfo{title}{Analyzing Facebook Privacy
  Settings: User Expectations vs. Reality}}.
\newblock In: {\sl \bibinfo{booktitle}{ACM SIGCOMM}}, \bibinfo{series}{IMC
  '11}, \bibinfo{publisher}{ACM}, pp. \bibinfo{pages}{61--70},
  \doi{10.1145/2068816.2068823}.

\bibitemdeclare{article}{lomuscio2015mcmas}
\bibitem{lomuscio2015mcmas}
\bibinfo{author}{Alessio \surnamestart Lomuscio\surnameend},
  \bibinfo{author}{Hongyang \surnamestart Qu\surnameend} \&
  \bibinfo{author}{Franco \surnamestart Raimondi\surnameend}
  (\bibinfo{year}{2017}): \emph{\bibinfo{title}{{MCMAS:} an open-source model
  checker for the verification of multi-agent systems}}.
\newblock {\sl \bibinfo{journal}{{STTT}}}
  \bibinfo{volume}{19}(\bibinfo{number}{1}), pp. \bibinfo{pages}{9--30},
  \doi{10.1007/s10009-015-0378-x}.

\bibitemdeclare{article}{PardoLic}
\bibitem{PardoLic}
\bibinfo{author}{Ra\'ul \surnamestart Pardo\surnameend},
  \bibinfo{author}{Musard \surnamestart Balliu\surnameend} \&
  \bibinfo{author}{Gerardo \surnamestart Schneider\surnameend}
  (\bibinfo{year}{2017}): \emph{\bibinfo{title}{Formalising privacy policies in
  social networks}}.
\newblock {\sl \bibinfo{journal}{Journal of Logical and Algebraic Methods in
  Programming}} \bibinfo{volume}{90}, pp. \bibinfo{pages}{125--157},
  \doi{10.1016/j.jlamp.2017.02.008}.

\bibitemdeclare{inproceedings}{PS14fpp}
\bibitem{PS14fpp}
\bibinfo{author}{Ra\'ul \surnamestart Pardo\surnameend} \&
  \bibinfo{author}{Gerardo \surnamestart Schneider\surnameend}
  (\bibinfo{year}{2014}): \emph{\bibinfo{title}{A Formal Privacy Policy
  Framework for Social Networks}}.
\newblock In: {\sl \bibinfo{booktitle}{SEFM'14}}, {\sl \bibinfo{series}{LNCS}}
  \bibinfo{volume}{8702}, \bibinfo{publisher}{Springer}, pp.
  \bibinfo{pages}{378--392}, \doi{10.1007/978-3-319-10431-7_30}.

\bibitemdeclare{techreport}{appendix}
\bibitem{appendix}
\bibinfo{author}{Ra\'ul \surnamestart Pardo\surnameend} \&
  \bibinfo{author}{Gerardo \surnamestart Schneider\surnameend}
  (\bibinfo{year}{2017}): \emph{\bibinfo{title}{Model Checking Social Network
  Models (Extended Version)}}.
\newblock \bibinfo{type}{Technical Report}, \bibinfo{institution}{Chalmers
  University of Technology}.
\newblock
  \urlprefix\url{http://www.cse.chalmers.se/~pardo/papers/model-checking-SNM-full-version.pdf}.

\bibitemdeclare{incollection}{RT11lkfsn}
\bibitem{RT11lkfsn}
\bibinfo{author}{Ji~\surnamestart Ruan\surnameend} \& \bibinfo{author}{Michael
  \surnamestart Thielscher\surnameend} (\bibinfo{year}{2011}):
  \emph{\bibinfo{title}{A logic for knowledge flow in social networks}}.
\newblock In: {\sl \bibinfo{booktitle}{IBERAMIA}},
  \bibinfo{publisher}{Springer}, pp. \bibinfo{pages}{511--520},
  \doi{10.1007/978-3-642-25832-9_52}.

\bibitemdeclare{article}{SLG13fel}
\bibitem{SLG13fel}
\bibinfo{author}{Jeremy \surnamestart Seligman\surnameend},
  \bibinfo{author}{Fenrong \surnamestart Liu\surnameend} \&
  \bibinfo{author}{Patrick \surnamestart Girard\surnameend}
  (\bibinfo{year}{2013}): \emph{\bibinfo{title}{Facebook and the Epistemic
  Logic of Friendship}}.
\newblock {\sl \bibinfo{journal}{CoRR}} \bibinfo{volume}{abs/1310.6440}.

\end{thebibliography}


\end{document}